\documentclass[journal]{IEEEtran}

\usepackage{amsmath}
\usepackage{tabularx}
\usepackage[numbers]{natbib}
\usepackage{graphicx}
\usepackage{algorithm}
\usepackage{algpseudocode}
\usepackage{subfigure}
\usepackage{enumerate}
\usepackage{amsfonts}
\usepackage{amsthm}
\usepackage{amssymb}
\usepackage{url}
\newtheorem{theorem}{Theorem}

\newtheorem{lemma}{Lemma}
\newtheorem{definition}{Definition}

\begin{document}

\title{\LARGE{Resilient and Decentralized Control of Multi-level Cooperative Mobile Networks to Maintain Connectivity under Adversarial Environment}} 
\vspace{-2mm}
\author{Juntao Chen and Quanyan Zhu \medskip
\\{~}
{ Department of Electrical and Computer Engineering, Tandon School of Engineering}\\{~}{New York University, Brooklyn, 11201, USA. E-mail: \{jc6412, qz494\}@nyu.edu.}\vspace{-4mm}}

\maketitle

\begin{abstract}               
Network connectivity plays an important role in the information exchange between different agents in the multi-level networks. In this paper,  we establish a game-theoretic framework to capture the uncoordinated nature of the decision-making at different layers of the multi-level networks. Specifically, we design a decentralized algorithm that aims to maximize the algebraic connectivity of the global network iteratively. In addition, we show that the designed algorithm converges to a Nash equilibrium asymptotically and yields an equilibrium network. To study the network resiliency, we introduce three adversarial attack models and characterize their worst-case impacts on the network performance. Case studies based on a two-layer mobile robotic network are used to corroborate the effectiveness and resiliency of the proposed algorithm and show the interdependency between different layers of the network during the recovery processes.
\end{abstract}



\section{Introduction}
Teams of mobile cooperative robots have a wide range of applications, such as rescue, monitoring, and searching in space exploration. One of the challenges in this kind of mobile robotic network (MRN) is to maintain the connectivity between robots, since a higher connectivity enables faster information spreading and hence a high level of situational awareness. Connectivity control of the MRN has been addressed in a number of previous works including \cite{michael2009maintaining,
kim2006maximizing,simonetto2011distributed} which have successfully tackled a single network of cooperative robots. Recent advances in networked systems have witnessed emerging applications involving multi-layer networks or \textit{network-of-networks} \cite{d2014networks,martin2014algebraic}. For example, when unmanned aerial vehicles (UAVs) and unmanned ground vehicles (UGVs) execute tasks together, the whole network can be seen as a two-layer interdependent network as shown in Fig. \ref{UAV}. Another example is the complex networks including mobile vehicular network and communication networks in public infrastructures. The interaction between mobile vehicles needs the support from communication network thus making two networks coupled. Therefore, the current single network control paradigm is not sufficient yet to address new challenges related to the analysis and design of multi-layer mobile networks.

\begin{figure}[t]
\centering
\includegraphics[width=1.8in]{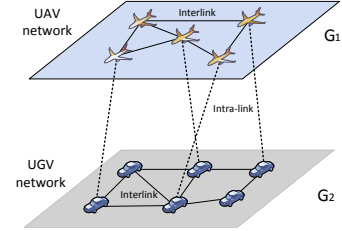}
\caption{Two-layer mobile robotic networks. One network consists of 5 UAVs, and the other network consists of 6 UGVs.}\label{UAV}\vspace{-5mm}
\end{figure}

The main objective of this work is to develop a theoretic framework that can capture the interactions between robots within a network and across networks. In our problem setting, each layer of the robotic network aims to maximize the connectivity of the overall network. If the whole network is fully cooperative or governed by a single agent, then the designed network is a \textit{team-optimal} solution. However, in practice, different layers of robotic networks are often operated by different entities, which makes the coordination between separate entities difficult. For example, in the previous UAV and UGV two-layer networks in Fig. 1, though the objectives for two networks are aligned, UAV is operated by the air force while UGV is operated by the army, which could lead to insufficient coordination. To address this problem, we establish a game-theoretic model in which two players control robots to maximize the global connectivity independently. This framework captures the lack of coordination between players in the multi-level networks. Furthermore, it guides the algorithmic design of decentralized mechanism for achieving an \textit{equilibrium} solution that is close to the team-optimal solution. In this paper, two players update their own robotic configuration based on the current one to maximize the network connectivity. This generates an iterative algorithm which converges to a Nash equilibrium (NE) point asymptotically and yields an equilibrium MRN.
   
An MRN is prone to adversarial attacks since a robot can be controlled by an adversary \cite{zengin2007real,kerns2014unmanned}, and the communication links between robots can be jammed \cite{xu2006jamming,
karlof2003secure,zhu2011eavesdropping}. Therefore, resilient control of the multi-level robotic networks to malicious attacks is critical to enhance its resiliency. To this end, we model the mobility of the robots by taking into account the communication links within and across the networks and use a game-theoretic approach to develop a resilient and decentralized algorithm for the individual networks. To study the network resiliency, we consider three attack models including the global positioning system (GPS) spoofing attack, targeted jamming attack and denial-of-service (DoS) attack. The impact of each attack can be quantified by measuring the difference of the algebraic connectivity of the network under a certain type of attack and without attacks. In addition, we characterize the worst-case of each attack and identify the interdependencies that exist in the multi-level networks. Case studies show that the robot removal resulting from the DoS attack will lead to the largest decrease in the network connectivity, and the MRN is the most resilient to the GPS spoofing attack which results in the constrained physical movement of robots by using the proposed control method. Furthermore, robots in the network without attacks will respond by moving to the positions that can set up the most  communication links with the robots in the attacked network during the recovery processes, and this shows the interdependency in the multi-level networks.

The contributions of this paper are summarized as follows:
\begin{enumerate}
\item We establish a game-theoretic framework that enables a decentralized control of mobile robots in the multi-level networks.
\item We introduce three attack models to the network and characterize their worst-case scenarios. In addition, we design a resilient and decentralized algorithm that aims to maximize the algebraic connectivity of the global MRN.
\item We show the convergence of the designed algorithm to a NE asymptotically, and corroborate its effectiveness and resiliency by case studies. The existence of interdependency in the multi-level MRN is also verified.

\end{enumerate}

\subsection{Related Work}
In the previous works, robotic network connectivity control problem is often addressed either in a totally centralized way  \cite{michael2009maintaining,kim2006maximizing,
ghosh2006growing,dai2011optimal}, or completely decentralized way  \cite{simonetto2011distributed,sabattini2013decentralized,
zavlanos2008distributed}. The centralized framework yields an optimal network but with low resiliency, since responses to failures in a global network may not be instantaneous. The network resulting from the decentralized framework is resilient; however, it is difficult to achieve the team solution with limited coordination. Our framework stands in between these two frameworks and thus leads to balanced features in terms of resiliency and optimality.

\subsection{Organization of the Paper}
The rest of the paper is organized as follows: Section~\ref{s1} presents some basics of graph theory and interdependent networks. We formulate a multi-level network formation game problem in Section~\ref{s2}. System dynamics discretization and the equilibrium solution concept is presented in Section \ref{s3}. A semidefinite programming approach and an iterative algorithm are proposed in Section~\ref{semi}. Section~\ref{s4} introduces three types of attacks to the MRN and characterizes their corresponding worst-case conditions. Case studies are given in Section~\ref{s5}, and Section~\ref{conclusion} concludes the paper.

\section{Background and Interdependent\\ Network Model}\label{s1}
Let $G(V,E)$ be an undirected graph composed by a set $V$ of $n$ nodes and a set $E$ of $m$ links with $n=|V|$, and $m=|E|$. For a link $e\in E$ connecting nodes $i$ and $j$ where the link weight is equal to $w_{ij}$, we define two vectors $\textbf{a}_e\in\mathbb{R}^n$ and $\textbf{b}_e\in\mathbb{R}^n$, where $\textbf{a}_e{(i)}=1$, $\textbf{a}_e{(j)}=-1$, $\textbf{b}_e{(i)}=w_{ij}$, $\textbf{b}_e{(j)}=-w_{ij}$, and other entries 0. Then, the Laplacian matrix $\textbf{L}$ of network $G$ can be expressed as 
\begin{equation}\label{laplacian2}
\textbf{L}=\displaystyle\sum_{e=1}^m{\textbf{a}_e}{\textbf{b}_e^T}.
\end{equation}
Basically, for the weighted Laplacian matrix $\textbf{L}$, its diagonal entries are equal to $\textbf{L}_{ii}=\sum_{j\in N_i}w_{ij}$, $\forall i \in V$, where $N_i$ denotes the set of nodes that are connected to node $i$. In addition, $\textbf{L}_{ij}=-w_{ij}$ if nodes $i$ and $j$ are connected, for $i\ne j\in V$, and 0 otherwise.  Note that $w_{ij}=w_{ji},\ \forall i,j\in V$.
By ordering the eigenvalues of $\textbf{L}$ in an increase way, we obtain
\begin{equation}
0=\lambda_1\le \lambda_2\le...\le \lambda_n,
\end{equation}
where $\lambda_2(\textbf{L})$ is called \textit{algebraic connectivity} of $G$. In addition, the \textit{Fiedler vector} of a network refers to the eigenvector associated with the eigenvalue $\lambda_2(\textbf{L})$ \cite{fiedler1973algebraic}.

For a two-layer interdependent network, we define two networks $G_1(V_1,E_1)$ and $G_2(V_2,E_2)$, where network 1 and network 2 are represented by $G_i$, for $i=1,2$, respectively. Network $i$, $i\in\{1,2\}$, is composed of $n_i=|V_i|$ nodes and $m_i=|E_i|$ links. The global network resulting from the connection of these two networks can be represented by $G=(V_1\ {\cup}\ V_2,\ E_1\ {\cup}\ E_2\ {\cup}\ E_{12})$, where $E_{12}$ is a set of \textit{intra-links} between $G_1$ and $G_2$. For convenience, we denote the network consisting of the intra-links between $G_1$ and $G_2$ as $G_{12}$. 
The adjacency matrix $\textbf{A}$ of the global network $G$ has the entry
$$ a_{ij}=\left\{
\begin{aligned}
w_{ij},\ &\mathrm{nodes}\ i\ \mathrm{and}\ j\ \mathrm{are\ connected};\\
0,\ &\mathrm{nodes}\ i\ \mathrm{and}\ j\ \mathrm{are\ not\ connected}.
\end{aligned}
\right.
$$
Let $\textbf{A}_1\in \mathbb{R}^{n_1\times n_1}$ and $\textbf{A}_2\in \mathbb{R}^{n_2\times n_2}$ be the adjacency matrices of $G_1$ and $G_2$, respectively, and $n=n_1+n_2$. When $E_{12}\ne{\emptyset}$, the adjacency matrix $\textbf{A}\in \mathbb{R}^{n\times n}$ takes the following form
\begin{equation*}
\textbf{A}=\begin{bmatrix}
\textbf{A}_1&\textbf{B}_{12}\\
\textbf{B}_{12}^T&\textbf{A}_2
\end{bmatrix},
\end{equation*}
where $\textbf{B}_{12}\in \mathbb{R}^{n_1\times n_2}$ is a matrix used to capture the effect of intra-links between networks. Define two diagonal matrices $\textbf{D}_1\in \mathbb{R}^{n_1\times n_1}$ and $\textbf{D}_2\in \mathbb{R}^{n_2\times n_2}$ as
$$
(\textbf{D}_1)_{ii}=\sum_{j}{(\textbf{B}_{12})_{ij}},\
(\textbf{D}_2)_{ii}=\sum_{j}{(\textbf{B}_{12}^T)_{ij}}.
$$
Then, by using
$
\textbf{L}=\textbf{D}-\textbf{A}, 
$
we obtain the Laplacian matrix
\begin{equation}\label{L}
\textbf{L}=\begin{bmatrix}
\textbf{L}_1+\textbf{D}_1&-\textbf{B}_{12}\\
-\textbf{B}_{12}^T&\textbf{L}_2+\textbf{D}_2
\end{bmatrix},
\end{equation}
where $\textbf{L}_1$ and $\textbf{L}_2$ are the Laplacians corresponding to $\textbf{A}_1$ and $\textbf{A}_2$, respectively.

\textit{\textbf{Remark 1}:} The above formulated two-layer interdependent framework can be easily extended to multi-layer cases.

\section{Problem Formulation}\label{s2}
In this section, we formulate a two-level mobile robotic network formation problem using a game-theoretic framework.

\subsection{Two-level Network Formation Game}\label{game_formulation}
The \textit{position} of robots in the network is denoted by a vector $\textbf{x}(t)=\big(x_1(t),\ x_2(t),...,x_n(t)\big)\in{\mathbb{R}^{3\times n}}$, and the dynamic of each robot $i$ is given by
$
\dot{x}_i(t)=u_i(t),
$
where $u_i(t)\in \mathbb{R}^{3}$ is the control of robot $i$ at time $t$. Besides, robots in the network can exchange data via wireless communications. Denote the communication link between robots $i$ and $j$ as $(i,j)$. Then, the strength of the communication link $(i,j)$ can be captured by the weight of the link. Thus, we can assign a weight function
$
w:\mathbb{R}^3\times \mathbb{R}^3\rightarrow \mathbb{R}_+
$
to each communication link $(i,j)$, such that
\begin{equation*}
w_{ij}(t)=w\big(x_i(t),x_j(t)\big)=g(\parallel
x_{ij}(t)\parallel_2),
\end{equation*}
for some $g:\mathbb{R}_+\rightarrow \mathbb{R}_+$, where $x_{ij}(t):=x_i(t)-x_j(t)$. The strength of a communication link decays exponentially with the distance \cite{tse2005fundamentals}. Therefore, the entries  $A_{ij}$ of the adjacency matrix $\textbf{A}$ admit
\begin{equation}\label{law}
A_{ij}=\left\{
\begin{aligned}
&1,\ &\parallel x_{ij}(t)\parallel_2<\rho_1;\\
&e^{\frac{-\alpha(\parallel x_{ij}(t)\parallel_2-\rho_1)}{\rho_2-\rho_1}},\ &\rho_1\le \parallel x_{ij}(t)\parallel_2\le \rho_2;\\
&0,\ &\parallel x_{ij}(t)\parallel_2> \rho_2,
\end{aligned}
\right.
\end{equation}
for $\rho_1,\rho_2\in\mathbb{R}_+$, and $\rho_1<\rho_2$. When the distance between the robots is less than $\rho_1$, the connectivity strength is up to 1; and when the distance is larger than $\rho_2$, robots lose the connection.

The model of a two-layer MRN is similar to the one in Fig. \ref{UAV}. Robots in the upper layer belong to network $G_1$, and robots in the bottom layer belong to network $G_2$. For convenience, we label robots in $G_1$ as $1,2,...,n_1\in V_1$, and robots in $G_2$ as $n_1+1,n_1+2,...,n_1+n_2\in V_2$. In addition, robots in the same layer and various layers can communicate with each other, and these communication links are called \textit{inter-links} and \textit{intra-links}, respectively. Note that exchanging data between robots in different layers is more difficult than that of the robots in the same layer due to much longer distance. Thus, to enable the information exchange between $G_1$ and $G_2$, we assume that the communication strength of intra-links has a larger value of $\rho_1$ and $\rho_2$ comparing with that of inter-links. 

For simplicity, define $-\gamma\triangleq \{1,2\}\setminus \gamma$, where $\gamma\in\{1,2\}$. We consider that two players, player 1 ($P_1$) and player 2 ($P_2$), play a network formation game. $P_1$ controls robots in network $G_1$, and $P_2$ controls robots in $G_2$. Specifically, $P_1$ and $P_2$ update their own mobile network iteratively by controlling the robots' positions which are denoted as $\textbf{x}_1$ and $\textbf{x}_2$, respectively. Note that $\textbf{x}_1:=(x_1,...,x_{n_1})\in\mathbb{R}^{3\times n_1}$, $\textbf{x}_2:=(x_{n_1+1},...,x_{n})\in\mathbb{R}^{3\times n_2}$, and $\textbf{x}:=(\textbf{x}_1,\textbf{x}_2)$. For each update, $P_\gamma$'s strategy is based on the current configuration of network $G_{-\gamma}$. The objectives of players are aligned by maximizing the algebraic connectivity of the whole network $G$, $\lambda_2\big(\textbf{L}_G(\textbf{x})\big)$, at every step. In addition, the action spaces of $P_1$ and $P_2$ are denoted by $\textbf{X}_1$ and $\textbf{X}_2$, respectively, which include all the possible network configurations. The set of pure strategy profiles $\textbf{X}:=\textbf{X}_1\times \textbf{X}_2$ is the Cartesian product of the individual pure strategy sets. Besides, the utility function for both players is $\lambda_2\big(\textbf{L}_G(\textbf{x})\big)$: $\textbf{X}\rightarrow \mathbb{R}_+$.

\textit{\textbf{Remark 2}:} In general, the objectives of two players can be different rather than maximize the algebraic connectivity of the global network. However, in our problem setting, the two teams of robots execute tasks collaboratively, and thus they both aim to optimize the global network connectivity to improve communications.

In the MRN formation game, one essential constraint is the minimum distance between the robots in the same layer. Without this constraint, all robots in the same layer will converge to one point finally which is unreasonable in reality. Thus, we assign a minimum distance for robots in $G_1$ and $G_2$ denoted by $d_1$ and $d_2$, respectively. Then, the network formation game can be represented by two individual \textit{interdependent} optimization problems  $Q_1^t$ and $Q_2^t$ as follows:
\begin{equation}\label{p1}
\begin{split}
Q_1^t:\ \ \ \max_{\textbf{x}_1(t+\tau_1)}&\ \ \lambda_2\big(\textbf{L}_G(\textbf{x}(t+\tau_1))\big)\\
\mathrm{s.t.}\ \ \ &||x_{ij}(t+\tau_1)||_2^2\ge d_1,\quad \ \forall i,j\in V_1,\\
&x_j(t+\tau_1)=x_j(t),\qquad \forall j\in V_2.
\end{split}
\end{equation}
\begin{equation}\label{p2}
\begin{split}
Q_2^t:\ \ \ \max_{\textbf{x}_2(t+\tau_2)}&\ \ \lambda_2\big(\textbf{L}_G(\textbf{x}(t+\tau_2))\big)\\
\mathrm{s.t.}\ \ \ &||x_{ij}(t+\tau_2)||_2^2\ge d_2,\quad \ \forall i,j\in V_2,\\
&x_j(t+\tau_2)=x_j(t),\qquad \forall j\in V_1,
\end{split}
\end{equation}
where $V_1$ and $V_2$ denote the sets of nodes in $G_1$ and $G_2$, respectively, and $\tau_1\in\mathbb{R}_+$ and $\tau_2\in\mathbb{R}_+$ are the time constants indicating the update frequency of the players. Note that $\tau_1$ and $\tau_2$ can be different because of distinct sensing, detection, and response capabilities of two mobile networks. Furthermore, smaller $\tau_1$ and $\tau_2$ indicate a higher resilience of the network to attacks, since a higher update frequency leads to faster system recovery.

\textit{\textbf{Remark 3}:} Note that the network formation game is played repeatedly over time, and its structure is the same only with different initial conditions in terms of the robots' position. In addition, at each stage of play, the game captured by $Q_1^t$ and $Q_2^t$ can be characterized as a constrained \textit{potential game} due to the identical objective of two players \cite{monderer1996potential}.


\section{System Dynamics Discretization and Equilibrium Solution Concept }\label{s3}
To address the MRN formation problem formulated in Section \ref{s2}, we first analyze the fundamental network algebraic connectivity maximization problem (ACMP). For a given network $G$ with $n$ nodes, its algebraic connectivity can be represented by 
\begin{equation}\label{lambdainf}
\lambda_2(\textbf{L}_G(\textbf{x}))=\min_{||z||_2=1,z\perp \textbf{1}}{z^T\textbf{L}_G(\textbf{x})z}
\end{equation}
based on the Courant-Fischer theorem \cite{horn2012matrix}, where \textbf{1} is an $n$-dimensional vector with all-one entries. In addition, the function $\lambda_2(\textbf{L}_G(\textbf{x}))$ is concave in $\textbf{L}_G(\textbf{x})$.
Therefore, the ACMP $\max_{\textbf{x}}\lambda_2(\textbf{L}_G(\textbf{x}))$
gives rise to convex optimization approaches to deal with our MRN control problem. For the unconstrained ACMP, we present the following theorem.

\begin{theorem}[\cite{nagarajan2012algorithms}]
 The network algebraic connectivity maximization problem $\max_{\textbf{x}}\lambda_2(\textbf{L}_G(\textbf{x}))$ is equivalent to the following:
\begin{equation}\label{thm1}
\begin{split}
&\max_{\textbf{x},\ \alpha}\ \ \alpha\\
&\ \mathrm{s.t.}\ \ \ \textbf{L}_G(\textbf{x})\succeq \alpha\cdot (\textbf{I}_{n}-\frac{1}{n}{\textbf{1}\textbf{1}}^T),
\end{split}
\end{equation}
where $\alpha\in\mathbb{R}$, and $\textbf{I}_n$ is an $n$-dimensional identity matrix.
\end{theorem}

\subsection{Discretization of the Dynamics}
For each update of the robotic network, it is essentially not continuous in time. Therefore, for simplicity, we discretize the formulated problems $Q_1^t$ and $Q_2^t$ in the following. First, we deal with the minimum distance constraint. By denoting $\mathcal{Z}_{ij}(t)=||x_{ij}(t)||_2^2$, and then differentiating $||x_{ij}(t)||_2^2$ with respect to the time, we obtain \cite{kim2006maximizing}
\begin{equation}\label{differentiate}
2\{\dot{x}_i(t)-\dot{x}_j(t)\}^T\{x_i(t)-x_j(t)\}=\dot{\mathcal{Z}}_{ij}(t).
\end{equation}
By using Euler's first order method
$
x(t)\rightarrow x(k),\ \ \dot{x}(t)\rightarrow \frac{x(k+1)-x(k)}{\Delta t},
$
where $\Delta t$ is the sample time, we rewrite \eqref{differentiate} as
\begin{equation}\label{discretedistance}
\begin{split}
2\{{x}_i(k+1)-{x}_j(k+1)\}^T\{x_i(k)-x_j(k)\}&\\
={\mathcal{Z}}_{ij}(k+1)+&{\mathcal{Z}}_{ij}(k).
\end{split}
\end{equation}

Similarly, differentiating and discretizing weight $w_{ij}$ yield
\begin{equation}\label{weight}
w_{ij}(k+1)=w_{ij}(k)+\frac{\partial f(||x_{ij}||_2)}{\partial ||x_{ij}||_2}\biggl\rvert_k \{{x}_{ij}(k+1)-{x}_{ij}(k)\}.
\end{equation}
Hence, we can obtain a discrete Laplacian matrix $\textbf{L}_G\big(\textbf{x}(k)\big)$ by using \eqref{weight} which is presented in Section \ref{reformulation}.

\subsection{Problem Reformulation}\label{reformulation}
Based on \eqref{thm1} and \eqref{discretedistance}, and for given initial position vectors $\textbf{x}_2(k)$ and $\textbf{x}_1(k)$ for $P_1$ and $P_2$, respectively, we can reformulate the problems $Q_1^t$ and $Q_2^t$ as follows:
\begin{equation}\label{ap1}
\begin{split}
\widetilde Q_1^k:\ \ \ &\max_{\textbf{x}_1(k+1),\ \alpha_1(k+1)}\quad \alpha_1(k+1)\\
\mathrm{s.t.}\ \ \ &\textbf{L}_G(k+1)\succeq \alpha_1(k+1)\cdot (I_{n}-\frac{1}{n}{\textbf{1}\textbf{1}}^T),\\
&2\{x_i(k+1)-x_j(k+1)\}^T\{x_i(k)-x_j(k)\}\\
&\qquad \ \ =\mathcal{Z}_{ij}(k+1)+\mathcal{Z}_{ij}(k),\\
&||x_{ij}(k+1)||_2^2\ge d_1,\ \ \ \forall i,j\in V_1,\\
&x_j(k+1)=x_j(k),\ \ \ \ \forall j\in V_2,\\
\end{split}
\end{equation}
\begin{equation}\label{ap2}
\begin{split}
\widetilde Q_2^k:\ \ \ &\max_{\textbf{x}_2(k+1),\ \alpha_2(k+1)}\quad\alpha_2(k+1)\\
\mathrm{s.t.}\ \ \ &\textbf{L}_G(k+1)\succeq \alpha_2(k+1)\cdot (I_{n}-\frac{1}{n}{\textbf{1}\textbf{1}}^T),\\
&2\{x_i(k+1)-x_j(k+1)\}^T\{x_i(k)-x_j(k)\}\\
&\qquad \ \ =\mathcal{Z}_{ij}(k+1)+\mathcal{Z}_{ij}(k),\\
&||x_{ij}(k+1)||_2^2\ge d_2,\ \ \ \forall i,j\in V_2,\\
&x_j(k+1)=x_j(k),\ \ \ \ \forall j\in V_1,
\end{split}
\end{equation}
where $\alpha_1(k+1)$ and $\alpha_2(k+1)$ are the scalar objectives of $\widetilde Q_1^k$ and $\widetilde Q_2^k$, respectively.

In addition, we obtain the discrete Laplacian matrix $\textbf{L}_G(k+1)$ by using \eqref{weight}, and its entries $l_{ij}^G (k+1)$ are
\begin{equation*}
\begin{split}
&l_{ij}^G (k+1)=\\
&\left\{
\begin{aligned}
&-w_{ij}(k+1),\quad\ \ \mathrm{if}\ i\neq j,\ (i,j)\in E_1\cup E_2;\\
&-\tilde w_{ij}(k+1),\quad \ \ \mathrm{if}\ i\in V_1,\ j\in V_2,\ \mathrm{or}\ j\in V_1,\ i\in V_2;\\
&\sum_{s\neq i,s\in V_1}\hspace{-2mm} w_{is}(k+1)+\sum_{q\neq i,q\in V_2}\hspace{-2mm} \tilde w_{iq}(k+1),\ \ \mathrm{if}\ \ i=j\in V_1;\\
&\sum_{s\neq i, s\in V_2}\hspace{-2mm} w_{is}(k+1)+\sum_{q\neq i, q\in V_1}\hspace{-2mm} \tilde w_{iq}(k+1),\ \ \mathrm{if}\ \ i=j\in V_2;
\end{aligned}
\right.
\end{split}
\end{equation*}
where $w_{ij},\ \forall (i,j)\in E_1\cup E_2$, represent the weight of interlinks inside $G_1$ and $G_2$, and $\tilde w_{ij},\ \forall (i,j)\in E_{12}$, denote the weight of intra-links connecting $G_1$ and $G_2$.

\subsection{Nash Equilibrium of the Game}
For the formulated discretized MRN formation game, a natural solution concept is Nash equilibrium (NE). Before presenting the formal definition of NE, we first analyze the impact of the players' action on the network at each step. Specifically, after $P_1$ takes his action at step $k$, $G_1$ and $G_{12}$ are reconfigured, where $G_{12}$ is the network between $G_1$ and $G_2$. We denote network $G_1$ and $G_{12}$ at stage $k$ as $G_{1,k}$ and $G_{12,k}$, respectively. For simplicity, we further define $\widetilde G_{12,k}:=G_{1,k} \cup G_{12,k}$, which is a shorthand notation for the merged network. Then, network $G_k$ can be expressed as $G_k= \widetilde G_{12,k}\ \cup G_{2,k}$. Similarly, after $P_2$ updates network $G_2$ at step $k$, the whole network $G_k$ becomes $ G_k=\widetilde G_{21,k}\ \cup G_{1,k}$, where $\widetilde G_{21,k}:=G_{2,k} \cup G_{12,k}$. Then, the formal definition of Nash equilibrium (NE) which depends on the \textit{position} of robots is as follows.


\begin{definition}[Nash Equilibrium]
The Nash equilibrium solution to the discretized multi-level robotic networks formation game is a strategy profile $\textbf{x}^*$, where $\textbf{x}^*=(\textbf{x}_1^*,{\textbf{x}_2^*})\in{\textbf{X}}$, that satisfy
\begin{equation*}\label{NE_discrete}
\begin{split}
\lambda_2\big(\textbf{L}_{ G_{k}}(\textbf{x}_1^*,\textbf{x}_2^*) \big)\ge \lambda_2\big(\textbf{L}_{G_{k}}(\textbf{x}_1,\textbf{x}_2^*) \big),\\
\lambda_2\big(\textbf{L}_{ G_{k}}(\textbf{x}_1^*,\textbf{x}_2^*) \big)\ge \lambda_2\big(\textbf{L}_{G_{k}}(\textbf{x}_1^*,\textbf{x}_2) \big),
\end{split}
\end{equation*}
for $\forall \textbf{x}_1\in \textbf{X}_1$ and $\forall \textbf{x}_2 \in \textbf{X}_2$, where $k$ denotes the time step, and $\textbf{x}=(\textbf{x}_1,{\textbf{x}_2})$ is defined in Section \ref{game_formulation}.
\end{definition}

Note that at the NE point, no player can individually increase the global network connectivity by reconfiguring their robotic network, and the two-level MRN possesses an equilibrium configuration.

\section{Semidefinite Programming and Iterative Algorithm}\label{semi}
In this section, we first derive a semidefinite programming (SDP) approach to address the discretized optimization problems 
$\widetilde Q_1^k$ and $\widetilde Q_2^k$, and then design an iterative algorithm to find the NE solution to the formulated MRN formation game.

\subsection{Semidefinite Programming Formulation}
Notice that in $\widetilde Q_1^k$ and $\widetilde Q_2^k$, the minimum distance constraints $||x_{ij}(k+1)||_2^2\geq d_1,\ \forall i,j\in V_1$, and  $||x_{ij}(k+1)||_2^2\geq d_2,\ \forall i,j\in V_2$, are \textit{nonconvex}. To address this issue, one method is to regard the distance $||x_{ij}(k+1)||_2^2=\mathcal{Z}_{ij}(k+1)$ as a new variable, and solve problems $\widetilde Q_1^k$ and $\widetilde Q_2^k$ with respect to unknowns $\mathcal{Z}_{ij}(k+1)$ and $\textbf{x}(k+1)$ jointly. In this way, $\widetilde Q_1^k$ and $\widetilde Q_2^k$ become convex problems. However, due to the coupling between the robots position and the distance vectors, solving $\widetilde Q_1^k$ and $\widetilde Q_2^k$ via merely adding new variables will yield inconsistency between the obtained solutions $\textbf{x}(k+1)$ and $\mathcal{Z}_{ij}(k+1),\ \forall i,j\in V$. Therefore, further considerations are needed, and we first present the definition of Euclidean distance matrix as follows. 
\begin{definition}[Euclidean Distance Matrix]
Given the positions of a set of $n$ points denoted by $\mathcal{N}:=\{x_1,...,x_n\}$, the Euclidean distance matrix representing the points spacing is defined as $$\textbf{D}:=[d_{ij}]_{i,j\in\mathcal{N}},\quad d_{ij}=||x_i-x_j||_2^2.$$
\end{definition}
A critical property of the Euclidean distance matrix is summarized in the following theorem.

\begin{theorem}[\cite{dattorro2010convex}]\label{EDM}
A matrix $\textbf{D}=[d_{ij}]_{i,j=1,...,n}$ is an Euclidean distance matrix if and only if 
\begin{equation}\label{sdp_EDM}
-\textbf{CDC}\succeq 0,\ \mathrm{and}\ d_{ii}=0,\ i=1,...,n,
\end{equation}
where $\textbf{C}:=\textbf{I}_n - \frac{1}{n}\textbf{11}^T$.
\end{theorem}

Note that \eqref{sdp_EDM} is a necessary and sufficient condition that ensures $\textbf{D}$ an  Euclidean distance matrix. In addition, the inequality and equality in \eqref{sdp_EDM} are both convex. Therefore, Theorem \ref{EDM} provides an approach to avoid the inconsistency between the robots position and distance vectors when they are treated as independent variables. In specific, denote $\textbf{Z}=[\mathcal{Z}_{ij}]_{i,j\in V}$, $\textbf{C}=\textbf{I}_n - \frac{1}{n}\textbf{11}^T$, and we can further reformulate problems $\widetilde Q_1^k$ and $\widetilde Q_2^k$ as
\begin{equation}
\begin{split}
\overline Q_1^k:\ \ \ &\max_{\textbf{x}_1(k+1),\ \textbf{Z}(k+1),\ \alpha_1(k+1)}\ \ \alpha_1(k+1)\\
\mathrm{s.t.}\ \ \ &\textbf{L}_G(k+1)\succeq \alpha_1(k+1) \textbf{C},\\
&2\{x_i(k+1)-x_j(k+1)\}^T\{x_i(k)-x_j(k)\}\\
&\qquad \ \ =\mathcal{Z}_{ij}(k+1)+\mathcal{Z}_{ij}(k),\\
&\mathcal{Z}_{ij}(k+1)\ge d_1,\ \ \ \ \ \forall i,j\in V_1,\\
& -\textbf{CZ}(k+1)\textbf{C}\succeq 0,\ \mathcal{Z}_{ii}(k+1)=0,\ i\in V,\\
&x_j(k+1)=x_j(k),\ \ \ \ \forall j\in V_2,\\
\end{split}
\end{equation}
\begin{equation}
\begin{split}
\overline Q_2^k:\ \ \ &\max_{\textbf{x}_2(k+1),\ \textbf{Z}(k+1),\ \alpha_2(k+1)}\ \ \alpha_2(k+1)\\
\mathrm{s.t.}\ \ \ &\textbf{L}_G(k+1)\succeq \alpha_2(k+1) \textbf{C},\\
&2\{x_i(k+1)-x_j(k+1)\}^T\{x_i(k)-x_j(k)\}\\
&\qquad \ \ =\mathcal{Z}_{ij}(k+1)+\mathcal{Z}_{ij}(k),\\
&\mathcal{Z}_{ij}(k+1)\ge d_2,\ \ \ \ \ \forall i,j\in V_2,\\
& -\textbf{CZ}(k+1)\textbf{C}\succeq 0,\ \mathcal{Z}_{ii}(k+1)=0,\ i\in V,\\
&x_j(k+1)=x_j(k),\ \ \ \ \forall j\in V_1.
\end{split}
\end{equation}

Hence, $\overline Q_1^k$ and $\overline Q_2^k$ become convex and are semidefinite programming problems which can be solved efficiently.

\subsection{Iterative Algorithm}
After obtaining the SDP problems $\overline Q_1^k$ and $\overline Q_2^k$, we aim to find the solution that results in an equilibrium network configuration. In the network formation game, $P_1$ controls robots in $G_1$ and reconfigures  the network by solving the optimization problem $\overline Q_1^k$ to obtain a new position of each robot. $P_2$  controls robots in network $G_2$ in a similar way by solving $\overline Q_2^k$. Note that the players' action at the current step can be seen as a best-response to the network at the previous step. Besides, the update frequency of each player in the discrete time measure is needed to be determined. For given $\tau_1$ and $\tau_2$ in the continuous time space, we can obtain their update frequencies by normalizing them into integers denoted by $s_1$ and $s_2$, respectively. Then, $P_1$ and $P_2$ reconfigure their robots for every $s_1$ and $s_2$ time intervals which can also be interpreted as the frequency of solving $\overline Q_1^k$ and $\overline Q_2^k$, respectively. Since both players maximize the global network connectivity at every update step, then one approach to find the equilibrium solution is to address $\overline Q_1^k$ and $\overline Q_2^k$ iteratively by two players until the yielding MRN possesses the same topology, i.e., $P_1$ and $P_2$ cannot increase the network connectivity further through relocating their robots.

\subsection{Feasibility and Convergence}
Before solving the problems $\overline Q_1^k$ and $\overline Q_2^k$, we should first analyze their feasibility, and we have the following lemma.
\begin{lemma}
For a given initial multi-level MRN where the distance between robots satisfies the predefined minimum distance constraint, then the game problems $\overline Q_1^k$ and $\overline Q_2^k$ are always feasible.
\end{lemma}

When $\overline Q_1^k$ and $\overline Q_2^k$ are feasible at each update step, another essential property is the convergence of the proposed iterative algorithm. Without loss of generality, we assume that two players will not update at the same step which can be easily achieved by normalizing the update frequency and choosing the initial update step of two players appropriately. Then, the convergence result is summarized in Theorem \ref{convergence}.

\begin{theorem}\label{convergence}
The iterative algorithm converges to a Nash equilibrium point asymptotically.
\end{theorem}
\begin{proof}
First, remind that both $\overline Q_1^k$ and $\overline Q_2^k$ maximize the algebraic connectivity of the global network, and thus the resulting $\alpha_i(k+1)$, $i\in\{1,2\}$, is no less than the one obtained from the previous update step which yields a non-decreasing network connectivity sequence $\lambda_2$. In addition, for a network with $n$ nodes, its algebraic connectivity is upper bounded by $n-1$, \cite{godsil2013algebraic}. Thus, based on the monotone convergence theorem \cite{ganter2012formal}, we can conclude that the network connectivity sequence converges asymptotically. Denote the actions of two players that achieve the network connectivity limit as $\bar{\textbf{x}}_1$ and $\bar{\textbf{x}}_2$ at some step $l$, and then, we obtain 
$
\lambda_2\big(\textbf{L}_{ G_{l}}(\bar{\textbf{x}}_1,\bar{\textbf{x}}_2) \big)\ge \lambda_2\big(\textbf{L}_{G_{l}}(\textbf{x}_1,\bar{\textbf{x}}_2) \big),\
\lambda_2\big(\textbf{L}_{ G_{l}}(\bar{\textbf{x}}_1,\bar{\textbf{x}}_2) \big)\ge \lambda_2\big(\textbf{L}_{G_{l}}(\bar{\textbf{x}}_1,\textbf{x}_2) \big),
$
for $\forall \textbf{x}_1\in \textbf{X}_1$ and $\forall \textbf{x}_2 \in \textbf{X}_2$. Otherwise, $\bar{\textbf{x}}_1$ and $\bar{\textbf{x}}_2$ do not result in the network connectivity limit. Obviously, the strategy pair $(\bar{\textbf{x}}_1,\bar{\textbf{x}}_2)$ satisfies the NE Definition \ref{NE_discrete} which indicates that the proposed iterative algorithm converges to a NE point asymptotically.
\end{proof}

\textit{\textbf{Remark 4}:} A typical example of the iterative algorithm is called \textit{alternating update} in which $P_1$ and $P_2$ have the same update frequency and reconfigure the MRN sequentially.

\section{Adversarial Attacks in the Networks}\label{s4}
Robots in the mobile networks are prone to malicious attacks \cite{xu2006jamming,
karlof2003secure, tseng2011survey,
khokhar2008review}. Thus, their secure and resilient control is essential. In this section, we first present three main types of adversarial attacks to the mobile network including the global positioning system (GPS) spoofing attack, targeted jamming attack and denial-of-service (DoS) attack, and then analyze their impacts on the network performance.

\subsection{GPS Spoofing Attack}
A GPS spoofing attack aims to deceive a GPS receiver in terms of the object's position, velocity and time by generating counterfeit GPS signals \cite{akos2012s}. In \cite{kerns2014unmanned}, the authors have demonstrated that UAVs can be controlled by the attackers and go to a wrong position through the GPS spoofing attack. In our MRN, we consider the scenario that the physical movement of a robot is constrained due to the attacks which can be realized by adding a disruptive position signal to the robot's real control command. Therefore, through the GPS spoofing attack, the mobile robot cannot move but still maintains its communications with other robots in the network. In addition, we assume that the attack cannot last forever but for a period of $g_a$ in the discrete time measure which is reasonable, since the resource of an attacker is limited, and the abnormal/unexpected behavior of the other unattacked robots resulting from the spoofing attack can be detected by the network administrator.

In the MRN, if robot $i$ is compromised by the spoofing attacker at time step $k_1$, and the attack lasts for $g_a$ time steps, then this scenario can be captured by adding the following constraint to the problems $\overline Q_1^k$ and $\overline Q_2^k$:
\begin{equation}
x_i(k+1)=x_i(k),\quad k=k_1,...,k_1+g_a-1.
\end{equation}

The attacked robot is usually randomly chosen. To evaluate the impact of the attack, we choose the robot that has the maximum degree which is denoted by ${i}_{max}$. Then, we obtain
\begin{align*}
{i}_{max}\in\ \mathrm{arg}\ \max_i\sum_{j\in N_i}w_{ij},
\end{align*}
where $N_i$ is the set of nodes that are connected to node $i$.

\subsection{Targeted Jamming Attack}
In wireless communication networks, one class of adversarial event is the jamming attack which can be launched by the attackers through injecting a huge amount of false data into the communication links \cite{karlof2003secure,zhu2011eavesdropping}. In this attack scenario, we consider the targeted jamming attack which means that the attacker jams a certain wireless communication channel between mobile robots which leads to a consequence of \textit{link removal} in the network.

To model this attack, denote the network as $\widetilde{G}(i,j)=(V,E\setminus(i,j))$ after removing a link $(i,j)\in E$ from network $G$, then, we have $\widetilde{\textbf{L}}=\textbf{L}-\Delta \textbf{L}$ and $\Delta \textbf{L}=\Delta \textbf{D}-\Delta \textbf{A}$, where $\Delta \textbf{D}$ and $\Delta \textbf{A}$ are the decreased degree and adjacency matrices, respectively. By using equation \eqref{laplacian2}, we obtain $\Delta \textbf{D}$ and $\Delta \textbf{A}$ as follows:
\begin{equation}
\begin{split}\label{delta}
\Delta \textbf{D}=\textbf{e}_i \tilde{\textbf{e}}_{i,j}^{T}+\textbf{e}_j \tilde{\textbf{e}}_{j,i}^T,\\
\Delta \textbf{A}=\textbf{e}_i \tilde{\textbf{e}}_{j,i}^T+\textbf{e}_j \tilde{\textbf{e}}_{i,j}^T,
\end{split}
\end{equation}
where $\textbf{e}_i$ and $\tilde{\textbf{e}}_{i,j}$ are zero vectors except the $i$-th element equaling to 1 and $w_{ij}$, respectively, and similar for $\textbf{e}_j$ and $\tilde{\textbf{e}}_{j,i}$. Denote the Laplacian matrix of $\widetilde{G}(i,j)$ as $\widetilde{\textbf{L}}(i,j)$, and by using equations in \eqref{delta}, we have
\begin{equation}\label{linkattackL}
\widetilde{\textbf{L}}(i,j)=\textbf{L}-\big(\textbf{e}_i-\textbf{e}_j\big)\big(\tilde{\textbf{e}}_{i,j}- \tilde{\textbf{e}}_{j,i}\big)^T.
\end{equation}

Similar to the GPS spoofing attack, the targeted jamming attack lasts for $g_b$ time steps. In order to incorporate this attack into the mobile robotic networks model, we add the following constraint to the Laplacian matrix:
\begin{equation}
w_{ij}(k)=0,\quad k= k_2,..., k_2+g_b-1,
\end{equation}
where $(i,j)$ denotes the attacked link, and $k_2$ is the starting point of the attack.

Attackers are often rational, i.e., they intentionally attack those communication links that whose removal will lead to the most decrease of the network connectivity. To characterize the worst-case of targeted jamming attack, we have the following analysis. When link $(i,j)$ is attacked, the resulting Laplacian is given by \eqref{linkattackL}. Denote the Fiedler vector of $\textbf{L}$ as $\textbf{u}$, and thus $\textbf{u}^T\textbf{Lu}=\lambda_2(\textbf{L})$ based on the definition. By using \eqref{lambdainf}, we can obtain the following:
\begin{equation}
\begin{split}
\lambda_2 \big(\widetilde{\textbf{L}}(i,j)\big)&\le \textbf{u}^T\widetilde{\textbf{L}}(i,j) \textbf{u}\\
&=\textbf{u}^T \big( \textbf{L}-\big(\textbf{e}_i-\textbf{e}_j\big)\big(\tilde{\textbf{e}}_{i,j}-\tilde{\textbf{e}}_{j,i}\big)^T \big) \textbf{u}\\
&=\textbf{u}^T\textbf{Lu}-(u_i-u_j)(w_{ij}u_i-w_{ji}u_j)\\
&=\lambda_2(\textbf{L})-w_{ij}(u_i-u_j)^2.
\end{split}
\end{equation}
Therefore, by removing the link $(i,j)^*$, where
\begin{equation}
(i,j)^*\in\mathrm{arg} \max_{(i,j)\in{E}}\ w_{ij}(u_i-u_j)^2,
\end{equation}
the upper bound for $\lambda_2 \big(\widetilde{\textbf{L}}(i,j)\big)$ is the smallest, and the algebraic connectivity of $G$ decreases the most.
 
\subsection{Denial-of-Service Attack}
In addition to the targeted link removal attack, another attack scenario corresponding to the wireless communications is the DoS attack \cite{wood2002denial}. The DoS attack can be realized by a number of technical methods including Wormhole, Blackhole and Grayhole attacks \cite{jhaveri2012attacks}. Specifically, in the MRN, the malicious attacker generates false message to flood the robots' communication resources which result in the \textit{node removal} of the network. When a node $i\in V$ is removed from the network $G$, then all links that are connected to node $i$ should also be removed. Denote the Laplacian matrix of the network after the attack as $\widetilde {\textbf{L}}(i)$, and remind that $N_i$ is the set of nodes that are connected to node $i$. Then, similar to the analysis of the link removal, we have
\begin{equation}\label{nodeattackL}
\widetilde{\textbf{L}}(i)=\textbf{L}-\sum_{j\in N_i}\big(\textbf{e}_i-\textbf{e}_j\big)\big(\tilde{\textbf{e}}_{i,j}- \tilde{\textbf{e}}_{j,i}\big)^T.
\end{equation}
If robot $i$ is attacked at time $k_3$, and the attack lasts for $g_c$ time steps, then, the following constraint is added to the Laplacian matrix:
\begin{equation}
w_{ij}(k)=0,\quad \forall j\in V_1 \cup V_2,\ k=k_3,..., k_3+g_c-1.
\end{equation}

In addition, the worst-case of denial-of-service attack can be captured as follows. When robot $i$ is attacked, the Laplacian of $G$ is changed to \eqref{nodeattackL}. Similar to the analysis of the most severe link removal attack, we obtain the following:
\begin{equation}
\begin{split}
\lambda_2 \big(\widetilde{\textbf{L}}(i)\big)&\le \textbf{u}^T\widetilde{\textbf{L}}(i) \textbf{u}\\
&=\textbf{u}^T \big( \textbf{L}-\sum_{j\in N_i}\big(\textbf{e}_i-\textbf{e}_j\big)\big(\tilde{\textbf{e}}_{i,j}-\tilde{\textbf{e}}_{j,i}\big)^T \big) \textbf{u}\\
&=\textbf{u}^T\textbf{Lu} - \sum_{j\in N_i}(u_i-u_j)(w_{ij}u_i-w_{ji}u_j)\\
&=\lambda_2(\textbf{L})-\sum_{j\in N_i} w_{ij}(u_i-u_j)^2.
\end{split}
\end{equation}

Hence, by attacking robot $i^*$, where
\begin{equation}
i^*\in \mathrm{arg}\ \max_{i} \sum_{j\in N_i}w_{ij} (u_i-u_j)^2,
\end{equation}
the algebraic connectivity of $G$ encounters the most decrease.

\textit{\textbf{Remark 5}:} Depending on the scope of knowledge that the attacker has of the network, our proposed game framework can be used for attackers of different knowledge levels. For example, an attacker may know the information of the whole multi-level network or merely one sub-network. For the former case of attack, closed form solutions have already been presented above. The analysis  for the latter case can be done in a similar way by focusing on a smaller network space.

\section{Case Studies}\label{s5}
In this section, we validate the obtained results via case studies. Specifically, we first show the performance of the two-level MRN by using the iterative algorithm. Then, we further quantify the impact of malicious attacks introduced in Section \ref{s4}, and assess the resiliency and interdependency of the network to malicious attacks.

\subsection{Effectiveness of the Algorithm}\label{effectiveness}
In the case studies, both networks $G_1$ and $G_2$ include 6 mobile robots, and the minimum distance is set to $0.4$. The link strength parameters of inter-links in two networks are the same, i.e., $\rho_1=1$, $\rho_2=3$ and $\alpha=5$, and for intra-links, the parameters are equal to $\rho_1=1.5$, $\rho_2=5$ and $\alpha=4$. Without loss of generality, $P_1$ and $P_2$ have the same update frequency and they reconfigure their robotic networks in an alternating fashion. In addition, YALMIP is adopted to solve the corresponding SDP problems \cite{lofberg2004yalmip}.  The obtained results of the MRN configuration trajectory without attack and its corresponding network algebraic connectivity are shown in Fig. \ref{configuration_without_attack} and Fig. \ref{without_attack_connectivity}, respectively. In specific, the MRN attains an equilibrium state after 10 updates which validates the effectiveness of the iterative algorithm.

\begin{figure}[t]
  \centering
  \subfigure[]{
    \label{configuration_without_attack} 
    \includegraphics[width=1.65in]{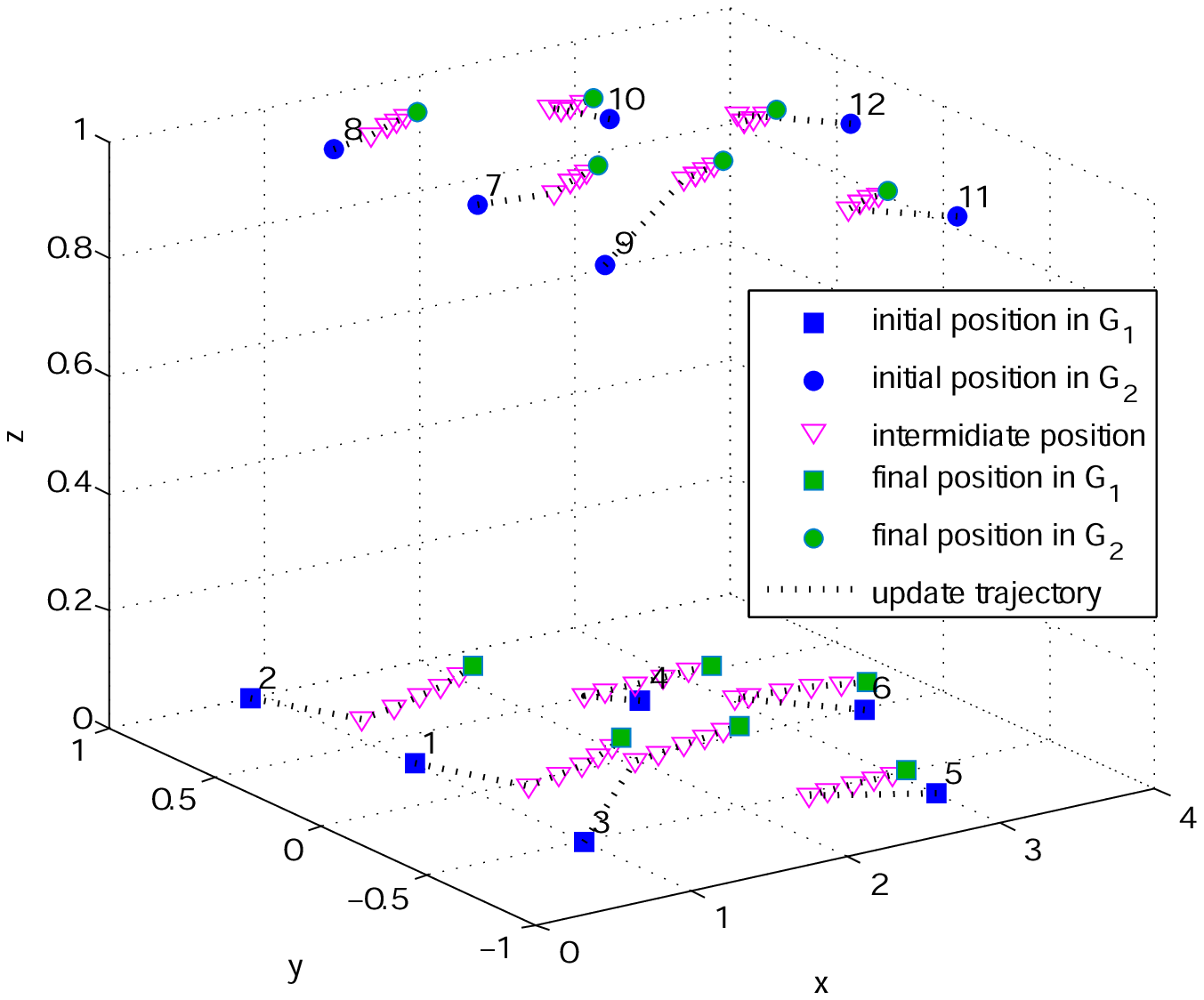}}
	 \subfigure[]{
    \label{without_attack_connectivity} 
    \includegraphics[width=1.6in]{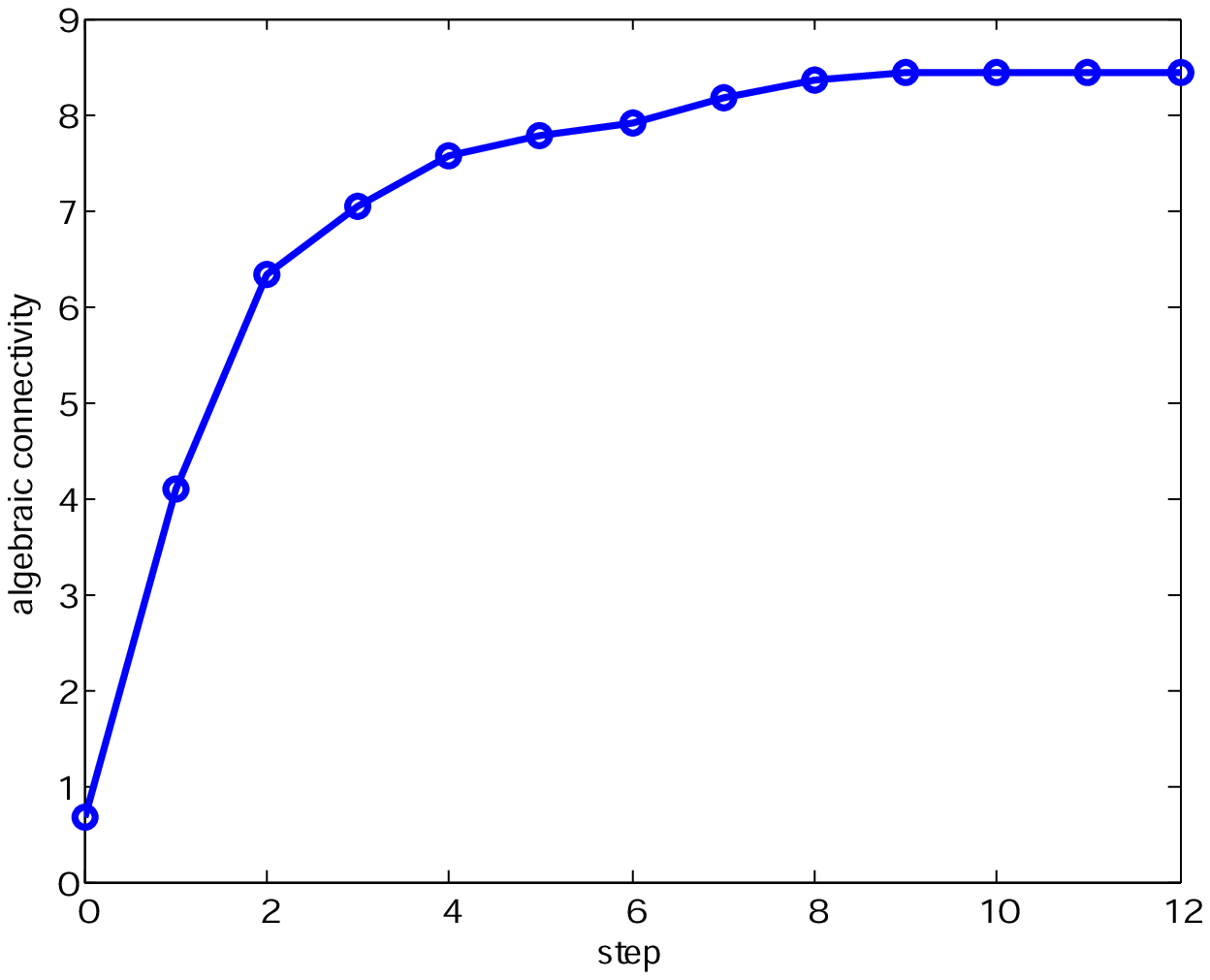}}
  \caption{(a) Configuration of a two-layer MRN without attack. (b) The resulting algebraic connectivity of the network formation game without attack.}
  \label{without_attack}
\end{figure}

\subsection{Impact of Malicious Attacks}\label{attacks}
In this section, we quantify the impact of each worst-case attack introduced in Section \ref{s4} on the network performance. For clarity, we assume that each attack is launched at step 6 during the network formation game, and without loss of generality, all attacks last for two steps. The result of network algebraic connectivity under each attack condition is shown in Fig. \ref{attack_connectivity}. We can see that the denial-of-service attack leads to the most decrease of the network connectivity comparing with other attacks, while the GPS spoofing attack is the least severe one.

\begin{figure}[t]
  \centering
  \subfigure[]{
    \label{attack_connectivity} 
    \includegraphics[width=1.65in]{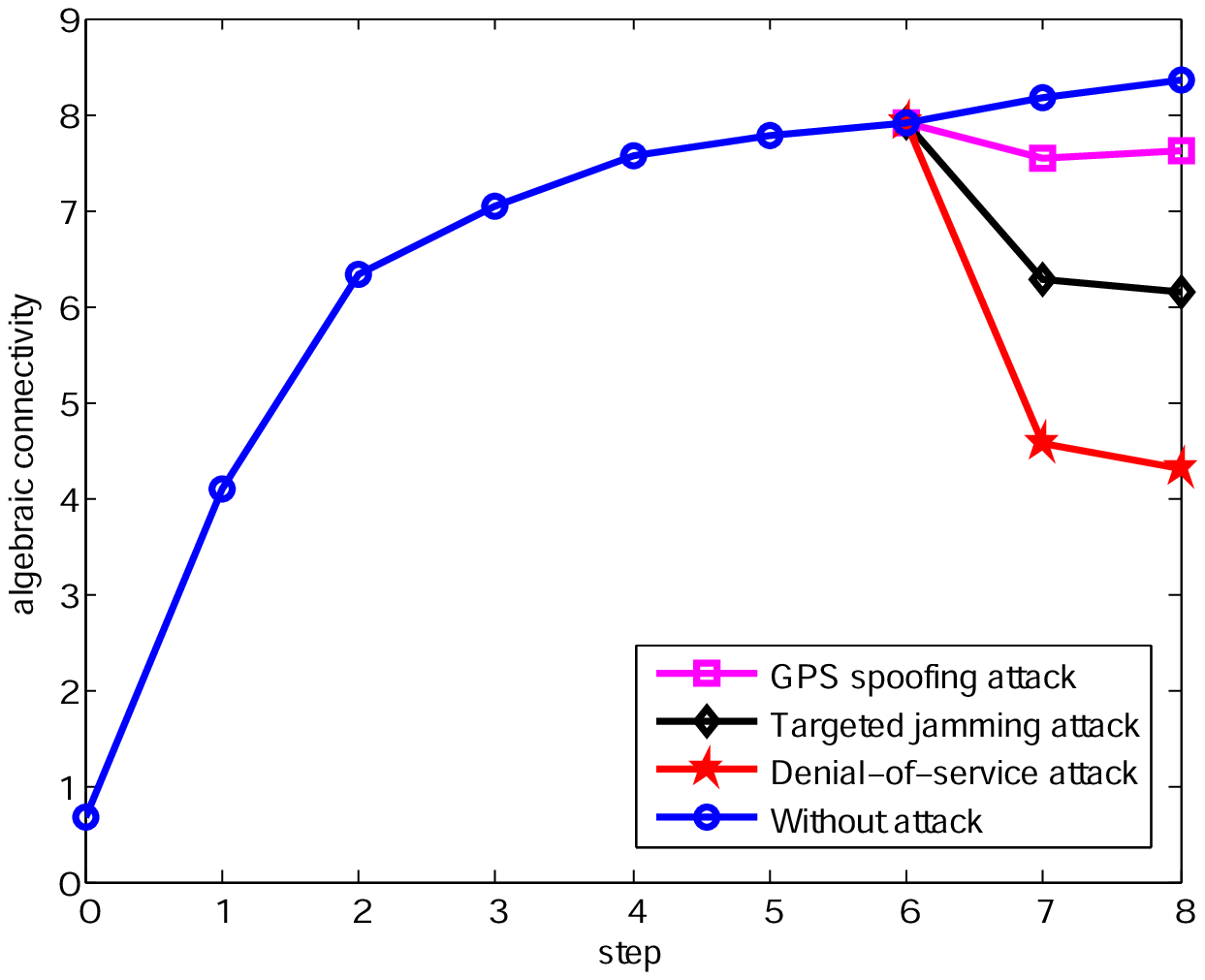}}
	 \subfigure[]{
    \label{attack_resiliency} 
    \includegraphics[width=1.65in]{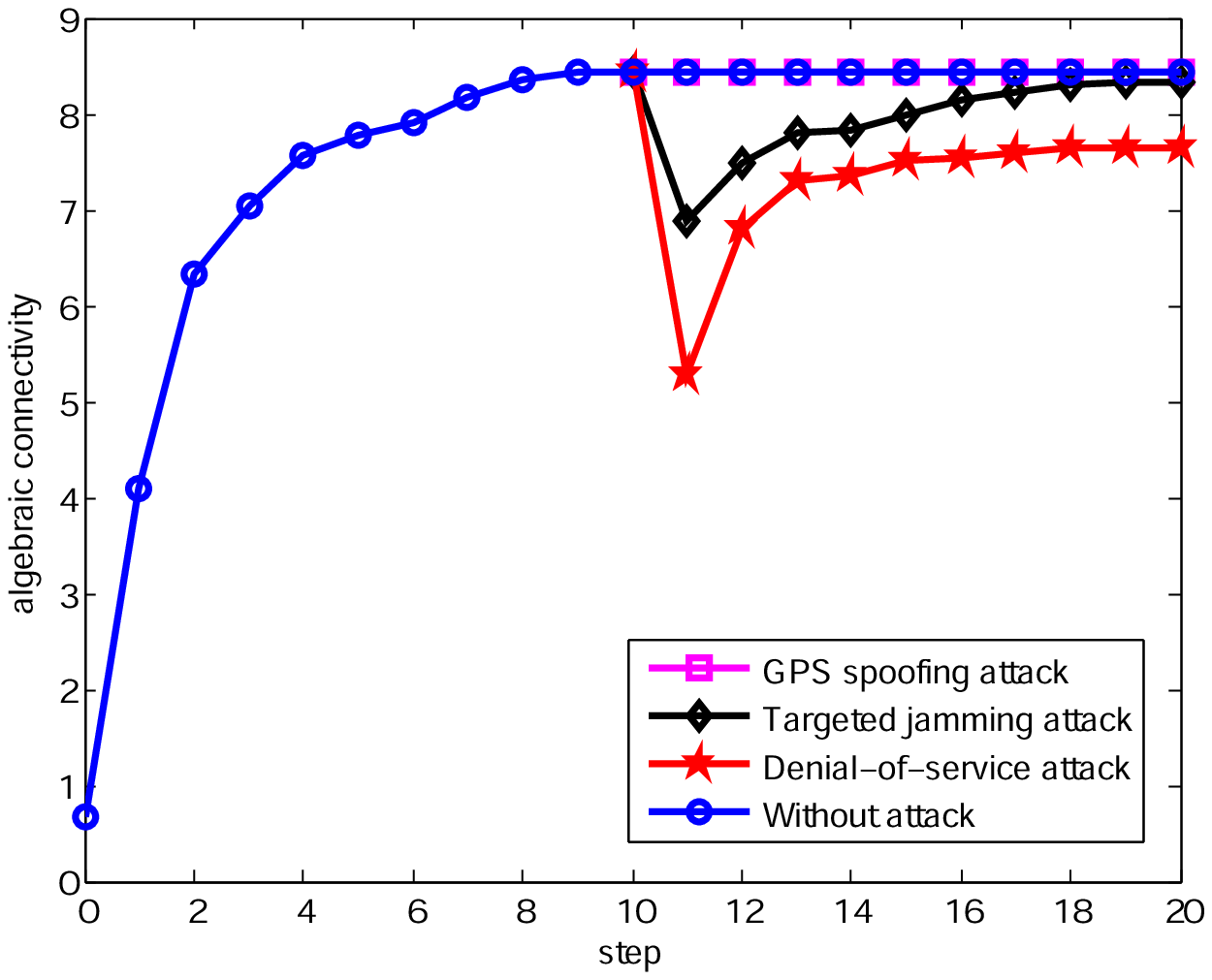}}
  \caption{(a) The impact of each worst-case attack on the network connectivity. (b) The algebraic connectivity of the network formation game without attack and under each attack condition. Under the adversarial environment, the worst-case attack happens at step 10, and it remains the same afterwards.}
\end{figure}

\subsection{Resiliency of the Network to Attacks}\label{resiliency}
After obtaining the impact of attacks on the network algebraic connectivity, the next step is to quantify the resiliency of the MRN to attacks. In specific, the resiliency metric is based on the system recovery speed and the recovery ability under the adversarial attacks. The adopted MRN model is the same as that in Section \ref{effectiveness}. In addition, we assume that the attacks are added to the MRN at step 10, and the attacker's action remains the same in the following steps.  Fig. \ref{attack_resiliency} shows the corresponding results. Specifically, the GPS spoofing attack does not impact the network performance in this case, since the attack is added at the point where MRN is of an equilibrium configuration, and the constrained physical movement of robots is not sufficient to decrease the network connectivity. For other cases, we can see that the MRN begins to recover after the attack happens which shows the high-level situational awareness of the MRN.  Moreover, besides the GPS spoofing attack, the MRN is the most resilient to the targeted jamming attack by using the designed iterative algorithm in terms of the agile recovery to a satisfying performance. The DoS attack can cause a huge loss of the algebraic connectivity, and the MRN cannot fully recover under this attack due to the removal of a robot. However, the rate of the network reaching a new NE is fast in this case.

\subsection{Interdependency of Multi-level Robotic Networks}\label{s6}
Comparing with single-level networks, a unique feature of the multi-level networks is their inherent interdependencies. We aim to show the existence of interdependency in the multi-level MRN in this section. Fig. \ref{interdependency_config} depicts the evolution of mobile network configuration corresponding to the DoS attack scenario in Section \ref{resiliency}. Remind that the attack happens at time step 10 where the network is under an equilibrium state. After introducing the attack, the two levels of robots will respond to it by moving to a new position, validating the resiliency of MRN.
In addition, by comparing two robotic networks in Fig. \ref{configuration_without_attack} and Fig. \ref{interdependency_config}, we find that robots in the network without attack (upper level) will move toward a position that can allow to set up the most intra-links with robots in the attacked network (lower level), and this fact corroborates the natural interdependency between two networks. Due to the interdependency, the whole mobile network can be more resilient to malicious attacks.

\begin{figure}[t]
\centering
\includegraphics[width=2.2in]{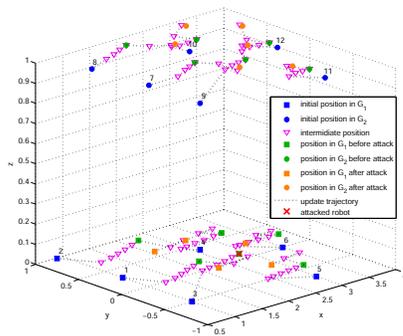}
\caption{MRN configuration under the DoS attack. The attack is introduced at step 10 where the initial network formation game reaches an equilibrium. Both mobile networks will respond to the DoS attack in the following update steps until reaching another equilibrium.}\label{interdependency_config}\vspace{-4mm}
\end{figure} 

\section{Conclusion}\label{conclusion}
In this paper, we have investigated the connectivity control of multi-level mobile robotic networks. We have developed a decentralized resilient algorithm to maximize the algebraic connectivity of the network to adversarial attacks, and shown its asymptotic convergence to a Nash equilibrium. Three types of attack models have been introduced, and their impacts have been quantified. Moreover, case studies have shown that the GPS spoofing attack has the least impact on the network performance, and the robotic network is the most resilient to the targeted jamming attack than other attacks by using the proposed control method. Future work can be designing a model predictive control algorithm that enables robots connectivity-aware during the network formation game.

\bibliographystyle{IEEEtran}
\bibliography{IEEEabrv,references}

\begin{thebibliography}{10}
\providecommand{\url}[1]{#1}
\csname url@samestyle\endcsname
\providecommand{\newblock}{\relax}
\providecommand{\bibinfo}[2]{#2}
\providecommand{\BIBentrySTDinterwordspacing}{\spaceskip=0pt\relax}
\providecommand{\BIBentryALTinterwordstretchfactor}{4}
\providecommand{\BIBentryALTinterwordspacing}{\spaceskip=\fontdimen2\font plus
\BIBentryALTinterwordstretchfactor\fontdimen3\font minus
  \fontdimen4\font\relax}
\providecommand{\BIBforeignlanguage}[2]{{%
\expandafter\ifx\csname l@#1\endcsname\relax
\typeout{** WARNING: IEEEtran.bst: No hyphenation pattern has been}%
\typeout{** loaded for the language `#1'. Using the pattern for}%
\typeout{** the default language instead.}%
\else
\language=\csname l@#1\endcsname
\fi
#2}}
\providecommand{\BIBdecl}{\relax}
\BIBdecl

\bibitem{michael2009maintaining}
N.~Michael, M.~M. Zavlanos, V.~Kumar, and G.~J. Pappas, ``Maintaining
  connectivity in mobile robot networks,'' in \emph{Experimental
  Robotics}.\hskip 1em plus 0.5em minus 0.4em\relax Springer, 2009, pp.
  117--126.

\bibitem{kim2006maximizing}
Y.~Kim and M.~Mesbahi, ``On maximizing the second smallest eigenvalue of a
  state-dependent graph laplacian,'' \emph{IEEE Transactions on Automatic
  Control}, vol.~51, no.~1, pp. 116--120, 2006.

\bibitem{simonetto2011distributed}
A.~Simonetto, T.~Keviczky, and R.~Babuska, ``On distributed maximization of
  algebraic connectivity in robotic networks,'' in \emph{American Control
  Conference, 2011}, 2011, pp. 2180--2185.

\bibitem{d2014networks}
G.~D'Agostino and A.~Scala, \emph{Networks of networks: the last frontier of
  complexity}.\hskip 1em plus 0.5em minus 0.4em\relax Springer, 2014, vol. 340.

\bibitem{martin2014algebraic}
J.~Mart{\'\i}n-Hern{\'a}ndez, H.~Wang, P.~Van~Mieghem, and G.~D’Agostino,
  ``Algebraic connectivity of interdependent networks,'' \emph{Physica A:
  Statistical Mechanics and its Applications}, vol. 404, pp. 92--105, 2014.

\bibitem{zengin2007real}
U.~Zengin and A.~Dogan, ``Real-time target tracking for autonomous uavs in
  adversarial environments: a gradient search algorithm,'' \emph{IEEE
  Transactions on Robotics}, vol.~23, no.~2, pp. 294--307, 2007.

\bibitem{kerns2014unmanned}
A.~J. Kerns, D.~P. Shepard, J.~A. Bhatti, and T.~E. Humphreys, ``Unmanned
  aircraft capture and control via gps spoofing,'' \emph{Journal of Field
  Robotics}, vol.~31, no.~4, pp. 617--636, 2014.

\bibitem{xu2006jamming}
W.~Xu, K.~Ma, W.~Trappe, and Y.~Zhang, ``Jamming sensor networks: attack and
  defense strategies,'' \emph{Network, IEEE}, vol.~20, no.~3, pp. 41--47, 2006.

\bibitem{karlof2003secure}
C.~Karlof and D.~Wagner, ``Secure routing in wireless sensor networks: Attacks
  and countermeasures,'' \emph{Ad hoc networks}, vol.~1, no.~2, pp. 293--315,
  2003.

\bibitem{zhu2011eavesdropping}
Q.~Zhu, W.~Saad, Z.~Han, H.~V. Poor, and T.~Ba{\c{s}}ar, ``Eavesdropping and
  jamming in next-generation wireless networks: A game-theoretic approach,'' in
  \emph{IEEE Military Communications Conference}, 2011, pp. 119--124.

\bibitem{ghosh2006growing}
A.~Ghosh and S.~Boyd, ``Growing well-connected graphs,'' in \emph{IEEE
  Conference on Decision and Control}, 2006, pp. 6605--6611.

\bibitem{dai2011optimal}
R.~Dai and M.~Mesbahi, ``Optimal topology design for dynamic networks,'' in
  \emph{IEEE Conference on Decision and Control and European Control
  Conference}, 2011, pp. 1280--1285.

\bibitem{sabattini2013decentralized}
L.~Sabattini, N.~Chopra, and C.~Secchi, ``Decentralized connectivity
  maintenance for cooperative control of mobile robotic systems,''
  \emph{International Journal of Robotics Research}, vol.~32, no.~12, pp.
  1411--1423, 2013.

\bibitem{zavlanos2008distributed}
M.~M. Zavlanos and G.~J. Pappas, ``Distributed connectivity control of mobile
  networks,'' \emph{IEEE Transactions on Robotics}, vol.~24, no.~6, pp.
  1416--1428, 2008.

\bibitem{fiedler1973algebraic}
M.~Fiedler, ``Algebraic connectivity of graphs,'' \emph{Czechoslovak
  Mathematical Journal}, vol.~23, no.~2, pp. 298--305, 1973.

\bibitem{tse2005fundamentals}
D.~Tse and P.~Viswanath, \emph{Fundamentals of wireless communication}.\hskip
  1em plus 0.5em minus 0.4em\relax Cambridge university press, 2005.

\bibitem{monderer1996potential}
D.~Monderer and L.~S. Shapley, ``Potential games,'' \emph{Games and economic
  behavior}, vol.~14, no.~1, pp. 124--143, 1996.

\bibitem{horn2012matrix}
R.~A. Horn and C.~R. Johnson, \emph{Matrix analysis}.\hskip 1em plus 0.5em
  minus 0.4em\relax Cambridge university press, 2012.

\bibitem{nagarajan2012algorithms}
H.~Nagarajan, S.~Rathinam, S.~Darbha, and K.~Rajagopal, ``Algorithms for
  synthesizing mechanical systems with maximal natural frequencies,''
  \emph{Nonlinear Analysis: Real World Applications}, vol.~13, no.~5, pp.
  2154--2162, 2012.

\bibitem{dattorro2010convex}
J.~Dattorro, \emph{Convex optimization \& Euclidean distance geometry}.\hskip
  1em plus 0.5em minus 0.4em\relax Meboo Publishing, 2008.

\bibitem{godsil2013algebraic}
C.~Godsil and G.~F. Royle, \emph{Algebraic graph theory}.\hskip 1em plus 0.5em
  minus 0.4em\relax Springer Science \& Business Media, 2013, vol. 207.

\bibitem{ganter2012formal}
B.~Ganter and R.~Wille, \emph{Formal concept analysis: mathematical
  foundations}.\hskip 1em plus 0.5em minus 0.4em\relax Springer Science \&
  Business Media, 2012.

\bibitem{tseng2011survey}
F.-H. Tseng, L.-D. Chou, and H.-C. Chao, ``A survey of black hole attacks in
  wireless mobile ad hoc networks,'' \emph{Human-centric Computing and
  Information Sciences}, vol.~1, no.~1, pp. 1--16, 2011.

\bibitem{khokhar2008review}
R.~H. Khokhar, M.~A. Ngadi, and S.~Mandala, ``A review of current routing
  attacks in mobile ad hoc networks,'' \emph{International Journal of Computer
  Science and Security}, vol.~2, no.~3, pp. 18--29, 2008.

\bibitem{akos2012s}
D.~M. Akos, ``Who's afraid of the spoofer? gps/gnss spoofing detection via
  automatic gain control (agc),'' \emph{Navigation}, vol.~59, no.~4, pp.
  281--290, 2012.

\bibitem{wood2002denial}
A.~D. Wood and J.~A. Stankovic, ``Denial of service in sensor networks,''
  \emph{Computer}, vol.~35, no.~10, pp. 54--62, 2002.

\bibitem{jhaveri2012attacks}
R.~H. Jhaveri, S.~J. Patel, and D.~C. Jinwala, ``Do{S} attacks in mobile ad hoc
  networks: A survey,'' in \emph{IEEE International Conference on Advanced
  Computing \& Communication Technologies}, 2012, pp. 535--541.

\bibitem{lofberg2004yalmip}
J.~Lofberg, ``Y{ALMIP}: A toolbox for modeling and optimization in matlab,'' in
  \emph{IEEE International Symposium on Computer Aided Control Systems Design},
  2004, pp. 284--289.

\end{thebibliography}

\end{document}